\documentclass{article}
\usepackage[margin=1.4in]{geometry}
\usepackage[utf8]{inputenc}

\usepackage{amsmath}
\usepackage{amssymb}
\usepackage{amsthm}
\usepackage{amstext}

\DeclareMathOperator*{\argmax}{arg\,max}

\usepackage{graphicx}
\usepackage{sidecap}

\theoremstyle{plain}
\newtheorem{lemma}{Lemma}[section]

\newtheorem{theorem}[lemma]{Theorem}

\newtheorem{observation}[lemma]{Observation}

\newtheorem{example}[lemma]{Example}
\theoremstyle{definition}
\newtheorem{definition}[lemma]{Definition}
\theoremstyle{remark}

\usepackage{authblk}

\usepackage{biblatex}
\title{A Note on Optimal Fees for Constant Function Market Makers}
\author{Robin Fritsch}
\author{Roger Wattenhofer}
\affil{ETH Zürich\\ \texttt{\{rfritsch,wattenhofer\}@ethz.ch}}
\date{}

\begin{document}

\maketitle

\begin{abstract}
We suggest a framework to determine optimal trading fees for constant function market makers (CFMMs) in order to maximize liquidity provider returns.
In a setting of multiple competing liquidity pools, we show that no race to the bottom occurs, but instead pure Nash equilibria of optimal fees exist.
We theoretically prove the existence of these equilibria for pools using the constant product trade function used in popular CFMMs like Uniswap. We also numerically compute the equilibria for a number of examples and discuss the effects the equilibrium fees have on capital allocation among pools.
Finally, we use our framework to compute optimal fees for real world pools using past trade data.
\end{abstract}

\section{Introduction}
In the past year, blockchains capable of executing smart contracts and decentralized finance (DeFi) applications built on top of them have seen tremendous growth in use.
One particularly popular and novel kind of applications are \emph{constant function market makers} (CFMMs) such as Uniswap, Balancer or Curve.

These decentralized exchanges (DEXs) allow users to trade tokens in a fully decentralized and non-custodial manner.
Moreover, they constitute a completely new type of market design compared to traditional exchanges using central limit order books:
Instead of matching orders of traders to liquidity providers' orders in the order book, they let traders swap tokens directly with a smart contract that holds the reserves of the liquidity providers.
For any pair of tokens, liquidity providers can create a \emph{pool} and deposit reserves of both tokens into the contract.
Traders can then send an amount of one of the tokens to the contract and receive a certain amount of the other token in return. The exact amount received is determined by a trade function depending on the reserves in the contract.
In the case of Uniswap, the amount is chosen such that the product of the reserves in the pool stays constant.
For every trade, the trader pays a small trading fee which is distributed pro rata among all liquidity providers contributing to the pool.

With the trading volume of CFMMs growing rapidly and already rivaling that of many centralized exchanges, it becomes essential to better understand this novel kind of market.
Since the number of CFMMs and liquidity pools is also growing quickly, a particularly relevant aspect is to examine how these markets compete with each other.
How should the pools optimally set their trading fees in order to attract the maximal possible fee revenue?
And how should liquidity providers allocate their liquidity among the growing number of pools?
In this paper, we suggest a framework to answer these questions.

We study the optimal fee problem in a setting of multiple liquidity pools competing with each other for trade volume.
In such a scenario one clearly expects a race to the bottom: all pools successively lowering their fees until they are almost zero. We show however, that this is not the case. Instead, we find that pure Nash equilibria exist, i.e.\ configurations of optimal fees for which no pool can gain anything by changing its fees, neither by increasing nor decreasing them.
For pools using the constant product trade function, such as Uniswap, we theoretically prove the existence of such pure Nash equilibria.
We also numerically compute these equilibria for several examples.
Here we find that larger pools can charge higher fees than smaller pools in the equilibria.
However, we somewhat surprisingly find that in the equilibria smaller pools earn a higher amount of fees relative to their size compared to larger pools. In other words, the return on investment is better for smaller pools.
This means liquidity is not always incentivized to concentrate: current or new liquidity providers would prefer to invest in smaller pools or might even create a new pool.

Using past trade data, we can also estimate optimal fees for real world pools.
We find that equilibrium fees are significantly lower than today's standard fee of $0.3\%$ suggesting that competition should lead to lower fees in the future.
In particular, we see that the pools would profit from unilaterally lowering their fees.

In game theoretic terms, we study a continuous $n$-person game where the liquidity pools are the players and the fees they choose to charge are their strategies. The utility of each pool is given by the amount of fees it receives.
Due to the nature of CFMMs, we can calculate from the pool reserves and fees how traders should optimally execute their trades using all available pools.
This makes it is possible to exactly quantify the change in trading volume the pools experience when changing their fees.
Note it can be advantages to split a trade among multiple pools specially for balanced pools with similar fees.
In practice, aggregators such as 1inch or Matcha help traders find and execute these optimal trades.

\section{Related Work}
The concept of automated market makers, i.e.\ mechanisms that automated the process of providing liquidity to a market, has been around for quite some time, e.g.\ in form of the logarithmic market scoring rule (LMSR) \cite{hanson03} which is used in prediction markets.
The recently popularised decentralized exchanges such as Uniswap use a new type of automated market maker design called constant function market maker (CFMM) \cite{angeris20oracles}. 
General properties of these markets, in particular how they behave alongside traditional centralized exchanges, have been studied in \cite{angeris19} and \cite{aoyagi2021liquidity}.
Furthermore, \cite{angeris2020curvature} analyzes how the curvature of a CFMM, i.e.\ the choice the of the constant function, influences liquidity provider returns.
For a comparison between traditional limit order book systems and constant product market makers, in particular how the latter can be emulated with certain order book shapes, see \cite{young20}.

The question of finding the optimal fees for constant function market makers has been discussed using a different model in \cite{evans21}:
Here arbitrage trading between the CFMM and an external reference market is considered.
For asset prices following a geometric Brownian motion, it is proven that it is optimal from the liquidity provider's perspective to set the fees arbitrarily low but not equal to zero.
The returns and no-arbitrage prices of liquidity provider shares in CFMMs have also been studied in \cite{tassy2020growth} and \cite{evans20}.
In contrast to these works, we study "real", i.e.\ non-arbitrage trades (these are sometimes referred to as uninformed trades since the traders do not have an information advantage over liquidity providers as arbitrageurs do).
More precisely, we consider traders that simply want to execute a planned trade at the best possible price.

In this note, we briefly touch on the problem of finding the optimal trade across multiple CFMM pools.
This problem of optimally routing trades though a network of constant function market makers is studied in detail in \cite{danos20}.

\section{Model}
We model the problem separately for each trading pair of tokens. 
For a pair of tokens, we consider a number of CFMM pools competing with each other for trade volume.
All trades are uninformed, i.e.\ the trader simply wants to swap a certain amount of the \emph{source token} for the greatest possible amount of the \emph{target token} using all available pools.
We assume that all pools are \emph{balanced} before the trade.
This means that an arbitrarily small amount can be swapped for the same price in all pools.
This is a reasonable assumption as arbitragers will always keep the pools close to balanced. Small remaining imbalances (e.g.\ caused by fees preventing arbitrage) can be ignored as their effects will average out since trades will be made in both possible directions in the pool.

First, we consider a single trade from the source token to the target token of a certain size (in the source token). This trade is executed optimally using all available pools, i.e.\ split among the pools to maximize the amount of target tokens received.
Later we will extend the model by assuming the size of the next trade is chosen from a distribution of trade sizes. This could be the distribution of trades observed in the recent past.

We assume pool fees are charged as a fixed percentage of the trade size (this is the norm on centralized and decentralized exchanges). For each pool, we look for the optimal such percentage from the liquidity provider's perspective, i.e.\ the fee that maximizes the total amount of fees the pool collects from the trade(s).

Each pool has a trade function $r:\mathbb{R}_{\geq 0}\to \mathbb{R}_{\geq 0}$ associated with it. This indicates for any non-negative amount of source tokens, how many target tokens a trader can receive in exchange in this pool.
A trade function naturally satisfies $r(0) = 0$, is monotonically increasing and continuous.
Furthermore, it is concave as the marginal price is increasing, and bounded since the liquidity in the pool is finite.
For our calculations and examples, we use the constant product trade function used by Uniswap and other CFMMs \cite{uniswap_whitepaper}. However, all of this also works for other trade functions.
A Uniswap-like pool with reserves $(A,B)$ and a fee of $s$ has the trade function
\begin{equation*}
    r(x) = B - \frac{AB}{A+(1-s)x}.
\end{equation*}
Assume we are dealing with $n$ pools and let the $i$-th pool have reserves $(A_i,B_i)$ and a fee of $s_i$.
Since we assumed that the pools are balanced, we have $A_i/B_i=A_j/B_j$ for all $i,j\in [1,n]$.
W.l.o.g. we can choose the unit of measurement for the target token such that $A_i=B_i$ for all pools.
Hence, the $i$-th pool has the trade function
\begin{equation}\label{eq:r_i}
    r_i(x) = A_i - \frac{A_i^2}{A_i+(1-s_i)x}.
\end{equation}
In particular, we only need the numbers $A_i$ and $s_i$ to fully characterize a pool. We measure the \emph{size} of a pool in source tokens meaning we say that the size of the $i$-th pool size is $2A_i$.

Lastly, we do not consider any fees besides the pool fees. In particular, we do not consider blockchain transaction fees. This means there is no extra cost for trading with a greater instead of a smaller number of pools.
While transaction fees are currently relatively high and only negligible for large trades on some blockchains, technological improvements are expected to lower transaction fees in near future making this assumption reasonable.

\section{Optimal Trade}
We use the following example to illustrate the next steps.
\begin{example}\label{ex:two_pools}
For a pair of tokens, there are two pools of sizes 2,000,000 and 4,000,000, respectively. Both pools have a trading fee of $0.3\%$. A trader wants to swap 1,000 source tokens for as many target tokens as possible.
\end{example}
Figure \ref{fig:trade_returns} shows how many target tokens the trader receives when swapping a certain fraction of the trade in pool 1 and the remaining part in pool 2. In this example, it is optimal to execute $1/3$ of the trade in pool 1 and $2/3$ in pool 2.

\begin{figure}[h]
    \centering
    \includegraphics[width=\textwidth]{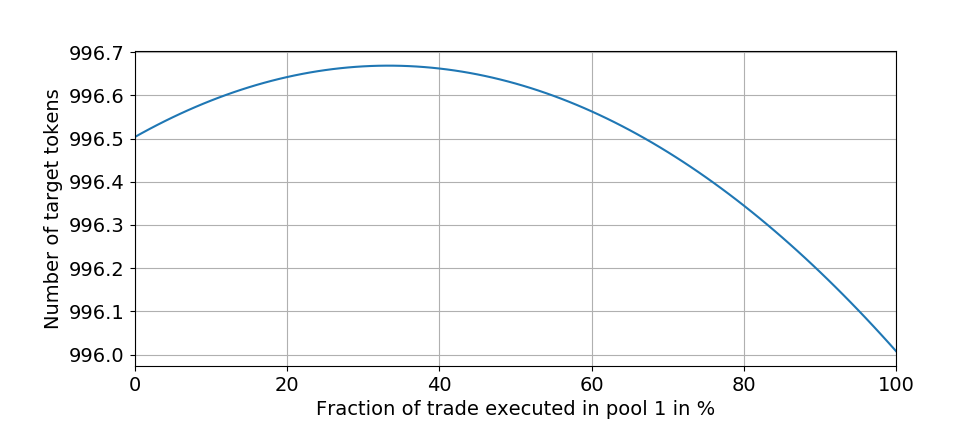}
    \caption{For Example \ref{ex:two_pools}, the graph shows how many target tokens the trader receives for 1000 source tokens when executing a certain fraction of the trade in pool 1 and the remaining part in pool 2.}
    \label{fig:trade_returns}
\end{figure}

Notice that it is always optimal split a trade across all pools if all fees are equal: Since we assumed the pools are balanced before the trade, the marginal price of an unused pool is always strictly lower than that of a used pool.
Hence, moving an arbitrarily small amount to an unused pool leads to a better trade.
Formally, the problem of finding the optimal trade can be stated as follows.

\begin{definition}[Optimal Trade Problem]
Given $n$ pools with trade functions $r_1,\ldots, r_n$ and a trade of size $t$, finding the optimal trade means solving

\begin{alignat}{3}
    &\text{maximize} & & r_1(x_1)+\ldots + r_n(x_1) \notag\\
    &\text{subject to} &\quad & x_1+\ldots +x_n = t \tag{OTP}\label{eq:OTP}\\
    & & & x_1, \ldots,x_n \geq 0\notag
\end{alignat}
\end{definition}

This problem can easily be solved numerically since all $r_i$ are naturally concave.
We use (multi-dimensional) ternary search to find the optimum. This technique does not required any specific knowledge about the trade function, in particular not their derivative.

Figure \ref{fig:split} shows the fraction of the trade that will optimally be traded in pool 1 for different fees in pool 1. (Note that while the graph looks piece-wise linear, it is not as the calculations below will show.)

\begin{figure}[h]
    \centering
    \includegraphics[width=\textwidth]{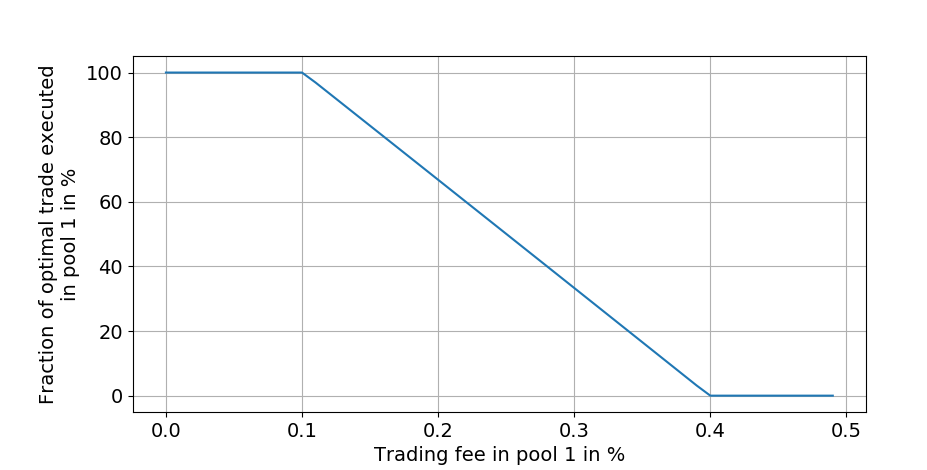}
    \caption{For Example \ref{ex:two_pools}, the graph shows which proportion of the optimal trade is executed in pool 1 for various fees of pool 1.}
    \label{fig:split}
\end{figure}

In order to prove the existence of Nash equilibria, we will solve \eqref{eq:OTP} analytically.
For our purpose, it suffices to solve the problem ignoring the constraint $x_1, \ldots,x_n \geq 0$.
To see this, note that in every Nash equilibrium the pools will choose fees such that the optimal trade satisfies the constraint $x_1, \ldots,x_n > 0$:
Independent of all other fees, every pool can always attract a positive fraction of the trade (and consequently a positive amount of fees) by setting its fees arbitrarily close to 0.
The derived formula will coincide with the actual solution shown in Figure \ref{fig:split} between about $0.1\%$ and $0.4\%$ while differing otherwise.

When ignoring the constraint, \eqref{eq:OTP} can be solved using a Lagrange multiplier. This yields the equations $r_i'(x_i)-\lambda = 0$ for $i=1,\ldots,n$ in addition to $x_1+\ldots +x_n = t$.
For the constant product trade function \eqref{eq:r_i}, the former equation becomes
\begin{equation*}
    \frac{A_i^2}{(A_i+(1-s_i)x_i)^2}(1-s_i)-\lambda = 0.
\end{equation*}
Solving this for $x_i$ yields
\begin{equation}\label{eq:x_i_lagrange}
    x_i = \frac{A_i}{\sqrt{1-s_i}}\frac{1}{\sqrt{\lambda}} - \frac{A_i}{1-s_i}.
\end{equation}
Summing over all $i$ leads to
\begin{equation*}
    t = \left( \sum_{j=1}^n \frac{A_j}{\sqrt{1-s_j}}\right)\frac{1}{\sqrt{\lambda}} - \sum_{j=1}^n \frac{A_j}{1-s_j}.
\end{equation*}
By solving this for $1/\sqrt{\lambda}$ and inserting the result back into \eqref{eq:x_i_lagrange}, we conclude
\begin{equation}\label{eq:x_i}
    x_i = \frac{A_i}{\sqrt{1-s_i}}\frac{\left( \sum_{j=1}^n \frac{A_j}{1-s_j}\right)+ t}{\sum_{j=1}^n \frac{A_j}{\sqrt{1-s_j}}} - \frac{A_i}{1-s_i}.
\end{equation}

It can be quickly checked that if the fees are equal in all pools, i.e.\ $s_1=s_2=\ldots=s_n$, the term in \eqref{eq:x_i} simplifies to $x_i = (A_i/\sum_{j=1}^n A_j)t$. This implies the following observation.

\begin{observation}
For Uniswap pools with identical fees, it is optimal for the trader to split the trade proportional to the pool sizes.
\end{observation}

\section{Optimal Fee Game}
We now study how the pools optimize their fees while competing with each other. In other words, we study how the pools play the \emph{optimal fee game}.
In game theoretic terms, we can think of the situation as a continuous $n$-person game in which each of the $n$ pools tries to maximize the amount of fees it collects.
More precisely, each pool chooses a strategy $s_i$ (its fee) from the set $S_i=[0,1]$ of possible strategies. We write $s = (s_1,\ldots,s_n)$, and denote $s_{-i}$ for the vector of the strategies of all pools except pool $i$, i.e.\ $s_{-i} = s/s_i$.

The $i$-th pool will then receive a trade of size $x_i$ according to \eqref{eq:x_i} leading to a total of $s_i x_i$ in fees.
So the utility function of player $i$ is given by $u_i(s) = s_i x_i$.
Figure \ref{fig:fee_amount} shows the amount of fees pool 1 collects depending on its fee. The optimal strategy for pool 1 in these circumstances is to set its fee to about $0.2\%$.

\begin{figure}[h]
    \centering
    \includegraphics[width=\textwidth]{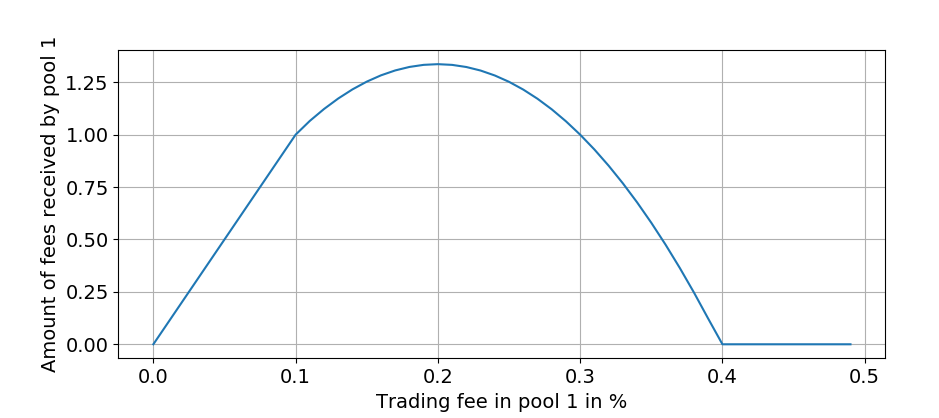}
    \caption{For Example \ref{ex:two_pools}, the graph shows the amount of fees pool 1 receives for various fees of pool 1.}
    \label{fig:fee_amount}
\end{figure}

As a response to pool 1 optimally adjusting its fees, pool 2 might also decide to optimize its fees. This could in turn change the optimal fee for pool 1 again.
This motivates us to find a configuration of fees such that none of the pools can gain anything from changing its fees, in other words: a pure Nash equilibrium.
Before numerically calculating these equilibria, we first theoretically prove their existence for the case of Uniswap-like pools.

\begin{theorem}
For any number $n\in \mathbb{N}$ pools using the constant product trade function and any trade size $t>0$, a Nash equilibrium exists in the optimal fee game.
\end{theorem}

\begin{proof}
We use the following sufficient condition due to Debreu \cite{debreu52}, Fan \cite{fan52} and Glicksberg \cite{glicksberg52}.
In a continuous $n$-person game, let $S_i$ be the strategy set and let $u_i$ be the utility function of the $i$-th player for $i=1,\ldots,n$. Then a pure Nash equilibrium exists if for every $i=1,\ldots,n$
\begin{itemize}
    \item $S_i$ is compact and convex,
    \item $u_i(s)$ is continuous in $s$, and
    \item $u_i(s_i, s_{-i})$ quasiconcave in $s_{i}$.
\end{itemize}
In our setting, the set of strategies $S_i=[0,1]$ is clearly compact and convex and $u_i$ is continuous. So it only remains to prove the quasiconcavity of $u_i$. (Remember that $u_i=s_i x_i$ with $x_i$ given by \eqref{eq:x_i}.) We consider a fixed $s_{-i}$ and write $u_i(s_i, s_{-i})$ as $u_i(s_i)$ for simplicity.
By introducing the constants (in $s_i$) 
\begin{align*}
    C = \left(\sum_{j\neq i} \frac{A_j}{1-s_j}\right) + x,\quad
    D = \sum_{j\neq i} \frac{A_j}{\sqrt{1-s_j}},
\end{align*}
we can write $u_i$ as
\begin{equation}\label{eq:u_i}
    u_i(s_i) = s_i\left(\frac{A_i}{\sqrt{1-s_i}} \frac{\frac{A_i}{1-s_i}+C}{\frac{A_1}{\sqrt{1-s_i}}+D} - \frac{A_i}{1-s_i}\right) = A_i s_i \frac{\frac{C}{D}-\frac{1}{\sqrt{1-s_i}}}{\frac{A_i}{D}+\sqrt{1-s_i}}.
\end{equation}
To prove that $u_i$ is quasiconcave, we need to show that $L(\alpha)=\{s_i\in[0,1] \mid u_i(s_i)\geq \alpha\}$ is convex for all $\alpha \in \mathbb{R}$.
This means we need to prove that $L(\alpha)$ is a single (possibly empty) interval.
First, note that the multiplicative constant of $A_i$ in \eqref{eq:u_i} can be ignored for this purpose.
Furthermore, let us introduce $C'=C/D$ and $D'=A_i/D$.
Finally, we substitute $t_i=\sqrt{1-s_i}$.
Since the inverse $s_i = 1-t_i^2$ is continuous and strictly decreasing on $[0,1]$, the set $L(\alpha)$ is an interval if and only if
\begin{equation*}
    L'(\alpha) = \left\{t_i\in [0,1] \Bigm\vert (1-t_i^2)\left(C'-\frac{1}{t_i}\right)\frac{1}{D'+t_i}\geq \alpha \right\}
\end{equation*}
is. The latter inequality is equivalent to
\begin{equation}\label{eq:t_i_ineq}
    f(t_i) := C't_i^3 + (\alpha-1)t_i^2 +(\alpha D'-C')t_i +1 \leq 0.
\end{equation}
Now we distinguish two cases.
If $\alpha D'-C'\geq 0$, then $\alpha>0$ and consequently $(\alpha-1)t_i^2+1>0$. Thus, the left-hand side of \eqref{eq:t_i_ineq} is always positive meaning the set $L'(\alpha)$ is empty.

Otherwise, if $\alpha D'-C'< 0$, consider $f$'s derivative 
\begin{equation}\label{eq:t_i_deriv}
    f'(t_i) = 3C't_i^2 + 2(\alpha-1)t_i + \alpha D'-C'.
\end{equation}
We see that $f(0)=1 \nleq 0$. So in order for $L'(\alpha)$ not to be a single interval, the function $f$ must attain the value 0 at least three times on $[0,1]$.
In particular, this implies that its derivative has two zeros in $[0,1]$. This together with \eqref{eq:t_i_deriv} means it can be written as $f'(t_i) = 3C'(t_i-z_1)(t_i-z_2)$ with $z_1,z_2\in [0,1]$. However, this would imply $\alpha D'-C' = 3C'z_1z_2\geq 0$ which is a contradiction.

\end{proof}

In order to find the Nash equilibria, we consider the \emph{best response function} of each pool. The best response function of a player returns the optimal strategy for this player for a given set of strategies of the other players.
Formally, it is defined as $b_i(s_{-i}) = \argmax_{s_i\in S_i} u_i(s_i,s_{-i})$.
A Nash equilibrium $s^*$ then satisfies $s_i^* = b_i(s_{-i}^*)$ for all $i=1,\ldots,n$.
We will numerically compute the equilibria by repeatedly applying $s_i = b_i(s_{-i})$ for $i=1,\ldots,n$ until $s$ converges.

\section{Equilibrium Fees}
We will now look at the fee equilibria for several examples.
Figure \ref{fig:two_pools} shows the equilibrium fees for two Uniswap-like pools with a varying difference in size.

\begin{figure}[h]
    \centering
    \includegraphics[width=\textwidth]{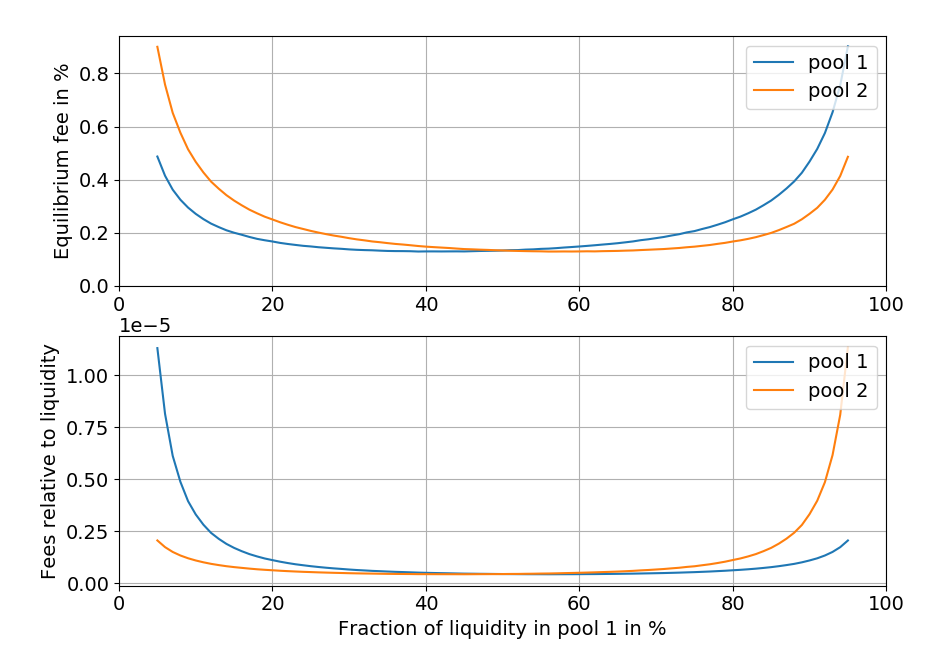}
    \caption{The two pools have a total size of 6,000,000 and the trade size is 1,000. The x-axis shows how much of the total liquidity is in pool 1, the rest is in pool 2. The upper graph plots the equilibrium fees while the lower graph shows the amount of fees the pool receives relative to its size.}
    \label{fig:two_pools}
\end{figure}

We see from the two curves in the upper graph that the larger of the two pools can always set a higher fee than the smaller one. That was to be expected as the larger pool has greater market power.
By only looking at the blue curve, we see a pool's equilibrium fee is high either when the pool is particularly large compared to its competitor or when it is particularly small. 
As one might expect, the lowest effective fee for the trader occurs when both pools have a similar size, i.e.\ when competition between the pools is "strongest".

Somewhat surprisingly, the lower graph shows that the smaller of the two pools earns a larger amount of fees relative to its size. So the liquidity providers providing liquidity to the smaller pool earn a higher return on their investment.
This incentivizes liquidity to move from the larger to the smaller pool. An equilibrium of this process is reached when both pools have the same size.

Next we consider three pools competing with each other. We vary the sizes of the first two pools while keeping the size of the third pools fixed.

\begin{figure}[h]
    \centering
    \includegraphics[width=\textwidth]{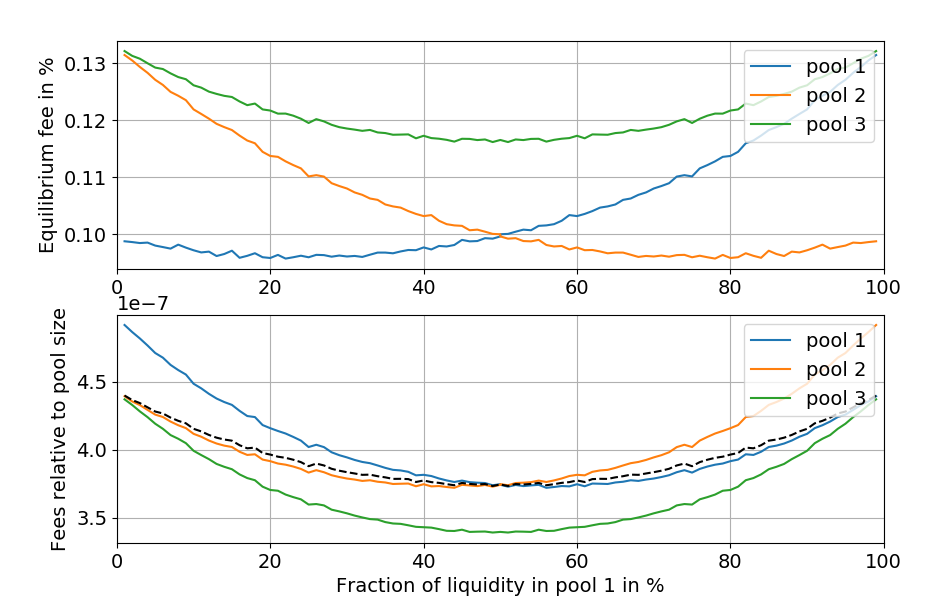}
    \caption{Pools 1 and 2 together have size 3,000,000. The x-axis indicates the part of this liquidity is in pool 1, the remaining part is in pool 2. Pool 3 has size 3,000,000.
    The dashed black line in the lower chart shows the (weighted) average of the relative fees in pools 1 and 2.}
    \label{fig:three_pools}
\end{figure}

Figure \ref{fig:three_pools} show that for three pools, we see similar effects as in the previous scenario with two pools: Larger pools can charge higher fees while smaller pools earn more fees relative to their liquidity.

Furthermore, we can compare two-pool and three-pool scenarios and reason about when it is advantages for one pool to split up into two new pools or for two pools to merge into one.
The relative fees of pools 2 and 3 for values on the $x$-axis close to 0 in Figure \ref{fig:three_pools} show the situation when we have two pools with size 3,000,000 each. 
We see that as long as the fraction of liquidity in pool 1 is between about 10\% to 90\%, the relative fees of pools 1 and 2 are both lower than the relative fees in the two-pool scenario.
So in these cases, it is beneficial for pools 1 and 2 to merge into a single pool.
In the remaining cases however, the smaller pool is better off alone. This also implies that for two equal-sized pools, it is beneficial for a small fraction of one of the pools to split off into a new pool.
On the other hand, the fact that the (weighted) average relative fee of pools 1 and 2 (the dashed black line) attains its maximum at 0 and 100 means that for the total liquidity in pools 1 and 2 as a whole it is always best to be combined in a single pool.

\section{Trade Size Distribution}
Obviously, in reality not all trades are of equal size. Accordingly, we now generalize our framework to more than a single trade size.
We assume trades are chosen from a discrete distribution, i.e.\ trade size $t_j$ occurs with a certain probability $p_j$ for $j\in J$.
This distribution could be derived from all trades observed in the recent past for a certain pair of tokens.
For a fee $s_i$ and a trade size $t_j$, let $x_i^*(s_i,t_j)$ be the solution of \eqref{eq:OTP}. Then the expected fees for pool $i$ are
\begin{equation}\label{eq:u_i_sum}
    u_i(s_i) = \sum_{j\in J} p_j s_i x_i^*(s_i,t_j).
\end{equation}
Using this utility function, we can calculate optimal fees from past trade data for real CFMM pools.
We consider the latest 1000 trades before block 12,000,000 (which occurred on 8 March 2021) from both the Uniswap and Sushiswap WETH-USDC pool, i.e.\ a sample of 2000 trades.
The liquidity of these pools stood at about \$262 million and \$401 million, respectively, at that time.
Assuming all these trades are executed optimally with the current $0.3\%$ fees, the Uniswap and Sushiswap pools would have received \$80,044 and \$122,051 in fees, respectively.

If the Sushiswap fee stays fixed, the best option for Uniswap would have been to lower its fee to $0.245\%$. This would have increased its earnings to \$102,892 with Sushiswap receiving \$76,225.
On the other hand, assuming Uniswap's fee remains untouched, Sushiswap's optimal fee choice would have been $0.27\%$ leading to \$53,285 and \$134,094 in earnings, respectively.

When considering more than one trade, it is no longer clear that a pure Nash equilibrium exists.
We can however use the average trades size of our sample to get an idea of the equilibrium-like fees.
Note that the formula for $x_i$ in \eqref{eq:x_i_lagrange} is actually linear in $t$, so using this formula with the utility function \eqref{eq:u_i_sum} is equivalent to simply using the average trade size. However, this is not completely accurate for multiple trade sizes: To arrive at the formula for $x_i$ we assumed that all pools attract a positive fraction of the trade in an equilibrium. But this is only true for a single trade at a time.
Nonetheless, taking the average trading size should give a reasonable estimate, in particular for relatively narrow distributions.
Using the average trade size, the equilibrium fee is $0.039\%$ and $0.045\%$ for Uniswap and Sushiswap, respectively. These fees would mean \$11,875 and \$16,867 in collected fees.

\section{Conclusion}
In this paper, we introduced a framework that allows us to determine optimal CFMM fees and reason about how the liquidity allocation among CFMM pools might evolve.
Of course, this does not answer the question of how to optimally set CFMM fees once and for all: Other aspects such as attracting fees from arbitrage trading (see related work) also need to be considered.
Also, trading fees for volatile tokens should be higher than for stable tokens to reward liquidity providers for the extra risk.
In any case, the situation of multiple pools for the same trading pair competing with each other is a relevant aspect of the problem. 
So the equilibria we find do give an indication how market forces could shape the fees.

\newpage

\printbibliography

\end{document}